\documentclass[12pt]{article}
\usepackage{bbding}
\usepackage{graphicx}
\usepackage{multirow}

\usepackage[T1]{fontenc}
\usepackage[sc]{mathpazo}
\usepackage{amsmath}
\usepackage{amssymb}
\usepackage{enumerate}
\usepackage{amsthm}
\usepackage{amsfonts,mathrsfs}

\usepackage{hyperref}
\hypersetup{pdfpagemode=UseNone}

\newcommand{\tinyspace}{\mspace{1mu}}

\newcommand{\norm}[1]{\left\lVert\tinyspace #1 \tinyspace\right\rVert}

\newcommand{\setft}[1]{\mathrm{#1}}

\newcommand{\density}[1]{\setft{D}\left(#1\right)}

\newcommand{\supp}{{\operatorname{supp}}}

\newenvironment{mylist}[1]{\begin{list}{}{
    \setlength{\leftmargin}{#1}
    \setlength{\rightmargin}{0mm}
    \setlength{\labelsep}{2mm}
    \setlength{\labelwidth}{8mm}
    \setlength{\itemsep}{0mm}}}
    {\end{list}}

\def\ot{\otimes}

\newcommand{\pa}[1]{(#1)}
\newcommand{\Pa}[1]{\left(#1\right)}

\newcommand{\Br}[1]{\left[#1\right]}

\newcommand{\Set}[1]{\left\{#1\right\}}

\DeclareMathOperator{\trace}{Tr}
\newcommand{\ptr}[2]{\trace_{#1}\pa{#2}}
\newcommand{\Ptr}[2]{\trace_{#1}\Pa{#2}}

\newcommand{\Tr}[1]{\Ptr{}{#1}}

\def\cH{\mathcal{H}}

\def\rS{\mathrm{S}}

\newtheorem{thrm}{Theorem}[section]

\newtheorem{prop}[thrm]{Proposition}
\newtheorem{cor}[thrm]{Corollary}
\theoremstyle{definition}

\numberwithin{equation}{section}

\newcounter{questionnumber}

\begin{document}

\title{On some entropy inequalities}

\author{Lin Zhang\footnote{E-mail: godyalin@163.com;
linyz@zju.edu.cn}\\
  {\small\it Institute of Mathematics, Hangzhou Dianzi University, Hangzhou 310018, PR~China}}

\date{}
\maketitle
\maketitle \mbox{}\hrule\mbox\\
\begin{abstract}

In this short report, we give some new entropy inequalities based on
R\'{e}nyi relative entropy and the observation made by Berta {\em et
al} [\href{http://arxiv.org/abs/1403.6102}{arXiv:1403.6102}]. These
inequalities obtained extends some well-known entropy inequalities.
We also obtain a condition under which a tripartite operator becomes
a Markov state.

\end{abstract}
\maketitle \mbox{}\hrule\mbox\\

\section{Introduction}

Recently, Carlen and Lieb \cite{Carlen} gives improvement of some
entropy inequalities by using \emph{Perels-Bogoliubov inequality}
and \emph{Golden Thompson inequality}. It is this paper that sparked
the present author to extend their work \cite{Lin} and get a
unifying treatment of some entropy inequalities via R\'{e}nyi
relative entropy \cite{Zhang}. Note that, by the monotonicity of
R\'{e}nyi relative entropy, the following inequality is derived:
\begin{eqnarray}
\rS(\rho||\sigma)\geqslant -2\log\Tr{\sqrt{\rho}\sqrt{\sigma}}
\end{eqnarray}
for two states $\rho,\sigma$. In fact, this inequality can be
extended as follows:
\begin{prop}
For a state $\rho\in\density{\cH}$ and a substate $\sigma$ on $\cH$
(i.e. $\Tr{\sigma}\leqslant1$), it holds that
\begin{eqnarray}\label{eq:substate}
\rS(\rho||\sigma)&\geqslant& -2\log\Tr{\sqrt{\rho}\sqrt{\sigma}}\\
&\geqslant&\norm{\sqrt{\rho}- \sqrt{\sigma}}^2_2\\
&\geqslant& \frac14\norm{\rho-\sigma}^2_1.
\end{eqnarray}
In particular $\rS(\rho||\sigma)=0$ if and only if $\rho=\sigma$.
\end{prop}

\begin{proof}
The following matrix inequality is important:
\begin{eqnarray}
\norm{\sqrt{M} - \sqrt{N}}^2_2 \leqslant\norm{M-N}_1\leqslant \norm{\sqrt{M} - \sqrt{N}}_2 \norm{\sqrt{M} + \sqrt{N}}_2,
\end{eqnarray}
where $M,N$ are positive matrices. Since both the traces of $\rho$
and $\sigma$ are no more than one, the desired inequality is
derived.
\end{proof}

\begin{prop}[\cite{Zhang}]
For two states $\rho,\sigma\in\density{\cH}$ and a quantum channel
$\Phi$ over $\cH$, it holds that
\begin{eqnarray}
&&\rS(\rho||\sigma) - \rS(\Phi(\rho)||\Phi(\sigma))
\\
&&\geqslant-2\log\Tr{\sqrt{\rho}\sqrt{\exp\Br{\log\sigma +
\Phi^*(\log\Phi(\rho)) - \Phi^*(\log\Phi(\sigma))}}}.
\end{eqnarray}
\end{prop}
Combining all of the above-mentioned inequality, we improved several
entropy inequalities , some of which simultaneously are obtained by
Carlen and Lieb \cite{Carlen}. In what follows, we list them here:
\begin{prop}[\cite{Carlen,Zhang}]
For two bipartite states $\rho_{AB},\sigma_{AB}\in
\density{\cH_{AB}}$ with $\cH_{AB}=\cH_A\ot\cH_B$, it holds that
\begin{eqnarray}
\rS(\rho_{AB}||\sigma_{AB}) - \rS(\rho_A||\sigma_A) &\geqslant& -2\log\Tr{\sqrt{\rho_{AB}}\sqrt{\exp(\log\sigma_{AB} - \log\sigma_A + \log\rho_A)}}\\
&\geqslant& \norm{\sqrt{\rho_{AB}} - \sqrt{\exp(\log\sigma_{AB} - \log\sigma_A + \log\rho_A)}}^2_2\\
&\geqslant& \frac14\norm{\rho_{AB} - \exp(\log\sigma_{AB} - \log\sigma_A + \log\rho_A)}^2_1.
\end{eqnarray}
\end{prop}
\begin{prop}[\cite{Lin,Zhang}]
\begin{eqnarray}\label{eq:lower-bound}
I(A:B|C)_\rho&\geqslant& -2\log\Tr{\sqrt{\rho_{ABC}}
\sqrt{\exp(\log\rho_{AC} - \log\rho_C + \log\rho_{BC})}}\\
&\geqslant& \norm{\sqrt{\rho_{ABC}} -
\sqrt{\exp(\log\rho_{AC} - \log\rho_C + \log\rho_{BC})}}^2_2\\
&\geqslant&\frac14 \norm{\rho_{ABC} -
\exp\Pa{\log\rho_{AC} + \log\rho_{BC} -\log\rho_C}}^2_1,
\end{eqnarray}
where $I(A:B|C)_\rho := \rS(\rho_{AC})+ \rS(\rho_{BC}) - \rS(\rho_{ABC}) -
\rS(\rho_C)$.
\end{prop}

Later Berta {\em et. al} \cite{Mark} present a R\'{e}nyi
generalization of quantum conditional mutual information
$I(A:B|C)_\rho$. We will employ some ideas from the paper
\cite{Mark} to derive some new entropy inequalities in this short
report. These inequalities obtained extends some well-known entropy
inequalities. We also obtain a condition under which a tripartite
operator becomes a Markov state, i.e. a state of vanishing
conditional mutual information.

Next, we give a brief introduction about the notation used here. We
consider only finite dimensional Hilbert space $\cH$. A
\emph{quantum state} $\rho$ on $\cH$ is a positive semi-definite
operator of trace one. The set of all quantum states on $\cH$ is
denoted by $\density{\cH}$. For each quantum state
$\rho\in\density{\cH}$, its von Neumann entropy is defined by
$\rS(\rho) := - \Tr{\rho\log\rho}$. The \emph{relative entropy} of
two mixed states $\rho$ and $\sigma$ is defined by
$$
\rS(\rho||\sigma) := \left\{\begin{array}{ll}
                             \Tr{\rho(\log\rho -
\log\sigma)}, & \text{if}\ \supp(\rho) \subseteq
\supp(\sigma), \\
                             +\infty, & \text{otherwise}.
                           \end{array}
\right.
$$
A \emph{quantum channel} $\Phi$ on $\cH$ is a trace-preserving
completely positive linear map defined over the set $\density{\cH}$.

The famous strong subadditivity (SSA) inequality of quantum
entropy, proved by Lieb and Ruskai in \cite{Lieb1973}, states that
\begin{eqnarray}\label{eq:SSA-1} \rS(\rho_{ABC}) +
\rS(\rho_C) \leqslant \rS(\rho_{AC}) + \rS(\rho_{BC}),
\end{eqnarray}
guaranteeing that $I(A:B|C)_\rho$ is nonnegative. Recently, the
operator extension of SSA is obtained by Kim in \cite{Kim2012}.
Following the line of Kim, Ruskai gives a family of new operator
inequalities in \cite{Ruskai2012}.

Ruskai is the first one to discuss the equality condition of SSA, that is, $I(A:B|C)_\rho = 0$. By
analyzing the equality condition of Golden-Thompson inequality, she
obtained the following characterization \cite{Ruskai2002}:
\begin{eqnarray}
I(A:B|C)_\rho = 0 \Longleftrightarrow \log\rho_{ABC} + \log\rho_C =
\log\rho_{AC} + \log\rho_{BC}.
\end{eqnarray}

Note here that conditional mutual information can be rewritten as
\begin{eqnarray}
I(A:B|C)_\rho = \rS(\rho_{ABC}||\exp(\log\rho_{AC} + \log\rho_{BC} - \log\rho_C)).
\end{eqnarray}
All we need to do is rewrite an involved quantity as a relative
entropy with the second argument being substate. Then
Prop.~\ref{eq:substate} is applied to get the desired inequality.

\section{Main results}

\begin{prop}[\cite{Mark}]
It holds that
\begin{eqnarray}
&&\rS(\rho_{ABC}||\exp(\log\sigma_{AC} + \log\tau_{BC} - \log\omega_C))\\
&&= I(A:B|C)_\rho + \rS(\rho_{AC}||\sigma_{AC}) + \rS(\rho_{BC}||\tau_{BC}) - \rS(\rho_C||\omega_C),
\end{eqnarray}
where $\rho_{ABC}\in\density{\cH_{ABC}}$, $\sigma_{AC} \in\density{\cH_{AC}},\tau_{BC}\in\density{\cH_{BC}}$, and $\omega_C\in\density{\cH_C}$.
\end{prop}

This identity leads to the following result:
\begin{prop}[\cite{Mark}]
It holds that
\begin{eqnarray}
&&\rS(\rho_{ABC}||\exp(\log\sigma_{AC} + \log\sigma_{BC} - \log\sigma_C))\\
&&= I(A:B|C)_\rho + \rS(\rho_{AC}||\sigma_{AC}) + \rS(\rho_{BC}||\sigma_{BC}) - \rS(\rho_C||\sigma_C),
\end{eqnarray}
where $\rho_{ABC}, \sigma_{ABC}\in\density{\cH_{ABC}}$.
\end{prop}
Using monotonicity of relative entropy, we have
$$
\rS(\rho_{AC}||\sigma_{AC}) \geqslant \rS(\rho_C||\sigma_C)~~\text{and}~~\rS(\rho_{BC}||\sigma_{BC}) \geqslant \rS(\rho_C||\sigma_C).
$$
This yields that
$$
\frac12\Br{\rS(\rho_{AC}||\sigma_{AC}) + \rS(\rho_{BC}||\sigma_{BC})} \geqslant \rS(\rho_C||\sigma_C).
$$
Therefore, we can draw the following conclusion:
\begin{thrm}\label{th:super-strong-additivity}
It holds that
\begin{eqnarray}
&&\rS(\rho_{ABC}||\exp(\log\sigma_{AC} + \log\sigma_{BC} - \log\sigma_C))\\
&&\geqslant I(A:B|C)_\rho + \frac12\rS(\rho_{AC}||\sigma_{AC}) + \frac12\rS(\rho_{BC}||\sigma_{BC}),
\end{eqnarray}
where $\rho_{ABC}, \sigma_{ABC}\in\density{\cH_{ABC}}$.
\end{thrm}

\begin{cor}
For two tripartite states $\rho_{ABC}, \sigma_{ABC}\in\density{\cH_{ABC}}$, it holds that
$$
\rS(\rho_{ABC}||\exp(\log\sigma_{AC} + \log\sigma_{BC} - \log\sigma_C)) \geqslant 0.
$$
In particular, $\rS(\rho_{ABC}||\exp(\log\rho_{AC} + \log\rho_{BC} -
\log\rho_C)) \geqslant 0$, i.e. $I(A:B|C)_\rho\geqslant0$, the
strong subadditivity inequality. Moreover
$\rS(\rho_{ABC}||\exp(\log\sigma_{AC} + \log\sigma_{BC} -
\log\sigma_C)) = 0$ if and only if $\rho_{ABC} =
\exp(\log\sigma_{AC} + \log\sigma_{BC} - \log\sigma_C)$.
\end{cor}

If $\rS(\rho_{ABC}||\exp(\log\sigma_{AC} + \log\sigma_{BC} - \log\sigma_C)) = 0$, then using Theorem~\ref{th:super-strong-additivity}, we have
\begin{eqnarray}
\begin{cases}
I(A:B|C)_\rho &= 0;\\
\rS(\rho_{AC}||\sigma_{AC}) &=0;\\
\rS(\rho_{BC}||\sigma_{BC}) &=0.
\end{cases}
\end{eqnarray}
This leads to the following:
\begin{eqnarray}
\rho_{AC} = \sigma_{AC},~~\rho_{BC} = \sigma_{BC}.
\end{eqnarray}
Thus $\rho_C=\sigma_C$. This indicates that
$$
\exp(\log\rho_{AC} + \log\rho_{BC} - \log\rho_C) = \exp(\log\sigma_{AC} + \log\sigma_{BC} - \log\sigma_C).
$$
Note that $I(A:B|C)_\rho = 0$ if and only if $\exp(\log\rho_{AC} + \log\rho_{BC} - \log\rho_C) = \rho_{ABC}$. Therefore $\exp(\log\sigma_{AC} + \log\sigma_{BC} - \log\sigma_C) = \rho_{ABC}$. From the above-mentioned process, it follows that
$$
\rS(\rho_{ABC}||\exp(\log\sigma_{AC} + \log\sigma_{BC} - \log\sigma_C))=0 \Longrightarrow \rho_{ABC} = \exp(\log\sigma_{AC} + \log\sigma_{BC} - \log\sigma_C).
$$
We know that, for any state $\sigma_{ABC}\in\density{\cH_{ABC}}$,
$$
\Tr{\exp(\log\sigma_{AC} + \log\sigma_{BC} - \log\sigma_C)} \leqslant 1.
$$
But what will happens if $\Tr{\exp(\log\sigma_{AC} + \log\sigma_{BC} - \log\sigma_C)} = 1$? In order to answer this question, we form an operator for any state $\sigma_{ABC}\in\density{\cH_{ABC}}$,
$$
\exp(\log\sigma_{AC} + \log\sigma_{BC} - \log\sigma_C).
$$
If $\exp(\log\sigma_{AC} + \log\sigma_{BC} - \log\sigma_C)$ is a valid state, denoted by $\rho_{ABC}$, then
$$
\rho_{AC} = \ptr{B}{\exp(\log\sigma_{AC} + \log\sigma_{BC} - \log\sigma_C)},~\rho_{BC} = \ptr{A}{\exp(\log\sigma_{AC} + \log\sigma_{BC} - \log\sigma_C)},
$$
and $\rho_C = \ptr{AB}{\exp(\log\sigma_{AC} + \log\sigma_{BC} - \log\sigma_C)}$. Furthermore
$\rS(\rho_{ABC}||\exp(\log\sigma_{AC} + \log\sigma_{BC} - \log\sigma_C))=0$. Thus $I(A:B|C)_\rho = 0$, i.e. $\exp(\log\sigma_{AC} + \log\sigma_{BC} - \log\sigma_C)$ is a Markov state.
\begin{thrm}
Given a state $\rho_{ABC}$. We form an operator $\exp(\log\rho_{AC} + \log\rho_{BC} - \log\rho_C)$. If
$$
\Tr{\exp(\log\rho_{AC} + \log\rho_{BC} - \log\rho_C)}=1,
$$
then the following statements are valid:
\begin{enumerate}[(i)]
\item $\exp(\log\rho_{AC} + \log\rho_{BC} - \log\rho_C) =
\rho^{1/2}_{AB}\rho^{-1/2}_B\rho_{BC}\rho^{-1/2}_B\rho^{1/2}_{AB}$;
\item $\exp(\log\rho_{AC} + \log\rho_{BC} - \log\rho_C) =
\rho^{1/2}_{BC}\rho^{-1/2}_B\rho_{AB}\rho^{-1/2}_B\rho^{1/2}_{BC}$.
\end{enumerate}
Therefore $\exp(\log\rho_{AC} + \log\rho_{BC} - \log\rho_C)$ must be
a Markov state.
\end{thrm}
From the above result, we see that if a state $\rho_{ABC}$ can be expressed by the form of $\exp(\log\sigma_{AC} + \log\sigma_{BC} - \log\sigma_C)$ for some state $\sigma_{ABC}$, then $\rho_{ABC}$ must ba a Markov state.

A question naturally arises: Which states $\rho_{ABC}$ are such that $\exp(\log\rho_{AC} + \log\rho_{BC} - \log\rho_C)$ is a Markov state? In other words, we are interested in the structure of the following set:
\begin{eqnarray}
\Set{\rho_{ABC}\in\density{\cH_{ABC}}: \Tr{\exp(\log\rho_{AC} + \log\rho_{BC} - \log\rho_C)}=1}.
\end{eqnarray}

\begin{thrm}
For two tripartite states $\rho_{ABC}, \sigma_{ABC}\in\density{\cH_{ABC}}$, it holds that
\begin{eqnarray}
&&\rS(\rho_{ABC}||\exp(\log\sigma_{AC} + \log\sigma_{BC} - \log\sigma_C))\\
&&\geqslant -2\log\Tr{\sqrt{\rho_{ABC}}\sqrt{\exp(\log\sigma_{AC} - \log\sigma_C + \log\sigma_{BC})}}\\
&&\geqslant\norm{\sqrt{\rho_{ABC}} - \sqrt{\exp(\log\sigma_{AC} - \log\sigma_C + \log\sigma_{BC})}}^2_2\\
&&\geqslant \frac14\norm{\rho_{ABC} - \exp(\log\sigma_{AC} - \log\sigma_C + \log\sigma_{BC})}^2_1.
\end{eqnarray}
\end{thrm}

\begin{proof}
Since $\Tr{\exp(\log\sigma_{AC} + \log\sigma_{BC} -
\log\sigma_C)}\leqslant1$, that is $\exp(\log\sigma_{AC} +
\log\sigma_{BC} - \log\sigma_C)$ is a substate, it follows from
\eqref{eq:substate} that the desired inequality is true.
\end{proof}

\begin{thrm}
For a tripartite state $\rho_{ABC}\in\density{\cH_{ABC}}$, it holds
that
\begin{eqnarray}
\rS(\rho_{AC}) + \rS(\rho_{BC}) - \rS(\rho_{ABC}) &\geqslant&
-2\log\Tr{\sqrt{\rho_{ABC}}\sqrt{\exp(\log\rho_{AC}+\log\rho_{BC})}}\\
&\geqslant&\norm{\sqrt{\rho_{ABC}} - \sqrt{\exp(\log\sigma_{AC} + \log\sigma_{BC})}}^2_2\\
&&\geqslant \frac14\norm{\rho_{ABC} - \exp(\log\sigma_{AC} +
\log\sigma_{BC})}^2_1.
\end{eqnarray}
\end{thrm}

\begin{proof}
All we need to do is rewrite $\rS(\rho_{AC}) + \rS(\rho_{BC}) -
\rS(\rho_{ABC})$ as a relative entropy:
$$
\rS(\rho_{AC}) + \rS(\rho_{BC}) - \rS(\rho_{ABC}) =
\rS(\rho_{ABC}||\exp(\log\rho_{AC}+\log\rho_{BC})).
$$
Next we prove that $\exp(\log\rho_{AC}+\log\rho_{BC})$ is a
substate: using Golden-Thompson inequality, we have
\begin{eqnarray}
\Tr{\exp(\log\rho_{AC}+\log\rho_{BC})}&\leqslant&
\Tr{\exp(\log\rho_{AC})\exp(\log\rho_{BC})}\\
&\leqslant& \Tr{\rho_{AC}\rho_{BC}} = \Tr{\rho^2_C}\leqslant 1.
\end{eqnarray}
This completes the proof.
\end{proof}

Further comparison with the inequalities in \cite{Kim2012,Ruskai2012} is left for the future research.


\subsubsection*{Acknowledgements}
This work is supported by NSFC (No.11301124).




\begin{thebibliography}{99}

\bibitem{Carlen}
E.A. Carlen, E.H. Lieb, {\em Remainder Terms for Some Quantum
Entropy Inequalities},
\href{http://arxiv.org/abs/1402.3840}{arXiv:1402.3840}

\bibitem{Lin}
L. Zhang,
{\em A lower bound of quantum conditional mutual informtaion},
\href{http://arxiv.org/abs/1403.1424v1}{arXiv:1403.1424}

\bibitem{Zhang}
L. Zhang,
{\em A stronger monotonicity inequality of quantum relative entropy: A unifying approach via R\'{e}nyi relative entropy},
\href{http://arxiv.org/abs/1403.5343v1}{arXiv:1403.5343}

\bibitem{Mark}
M. Berta, K. Seshadreesan, M. Wilde,
{\em R\'{e}nyi generalizations of the conditional quantum mutual information},
\href{http://arxiv.org/abs/1403.6102}{arXiv:1403.6102}


\bibitem{Lieb1973}
E. Lieb and M. Ruskai, {\em Proof of the strong subadditivity of
quantum-mechanical entropy}, J. Math. Phys.
\href{http://dx.doi.org/10.1063/1.1666274}{\textbf{14}, 1938-1941
(1973).}

\bibitem{Kim2012}
I. Kim, {\em Operator extension of strong subadditivity of entropy},
J. Math. Phys.
\href{http://dx.doi.org/10.1063/1.4769176}{\textbf{53}, 122204
(2012).}

\bibitem{Ruskai2012}
M. Ruskai, {\em Remarks on on Kim's strong subadditivity matrix
inequality: extensions and equality conditions}, J. Math. Phys.
\href{http://dx.doi.org/10.1063/1.4823581}{\textbf{54}, 102202
(2013).}

\bibitem{Ruskai2002}
M. Ruskai, {\em Inequalities for quantum entropy: A review with
conditions for equality}, J. Math. Phys.
\href{http://dx.doi.org/10.1063/1.1497701}{\textbf{43}, 4358-4375
(2002)}; erratum
\href{http://dx.doi.org/10.1063/1.1824214}{\textbf{46}, 019901
(2005)}.



\end{thebibliography}
\end{document}